\documentclass[11pt,letterpaper]{article}
\usepackage{amsmath}
\usepackage{amssymb}
\usepackage{graphicx}
\usepackage{fullpage}
\usepackage[amsthm,amsmath,thmmarks]{ntheorem}

\newtheorem{theorem}{Theorem}

\newtheorem{lemma}{Lemma}
\newtheorem{claim}{Claim}

\newtheorem{definition}{Definition}

\newcommand{\p}{{\rm P}}
\newcommand{\NP}{{\rm NP}}

\begin{document}

\title{Lift-and-Project Integrality Gaps for the Traveling Salesperson Problem}
\author{Thomas Watson\footnote{Computer Science Division, University of California, Berkeley. This material is based upon work supported by the National Science Foundation Graduate Research Fellowship under Grant No.~DGE-0946797 and by the National Science Foundation under Grant No.~CCF-1017403.}}
\maketitle

\begin{abstract}
We study the lift-and-project procedures of Lov{\'a}sz-Schrijver and Sherali-Adams applied to the standard linear programming relaxation of the traveling salesperson problem with triangle inequality. For the asymmetric TSP tour problem, Charikar, Goemans, and Karloff (FOCS 2004) proved that the integrality gap of the standard relaxation is at least $2$. We prove that after one round of the Lov{\'a}sz-Schrijver or Sherali-Adams procedures, the integrality gap of the asymmetric TSP tour problem is at least $3/2$, with a small caveat on which version of the standard relaxation is used. For the symmetric TSP tour problem, the integrality gap of the standard relaxation is known to be at least $4/3$, and Cheung (SIOPT 2005) proved that it remains at least $4/3$ after $o(n)$ rounds of the Lov{\'a}sz-Schrijver procedure, where $n$ is the number of nodes. For the symmetric TSP path problem, the integrality gap of the standard relaxation is known to be at least $3/2$, and we prove that it remains at least $3/2$ after $o(n)$ rounds of the Lov{\'a}sz-Schrijver procedure, by a simple reduction to Cheung's result.
\end{abstract}


\section{Introduction}
\label{tsp:sec:intro}

The traveling salesperson problem (TSP) is the following: Given a complete directed graph with nonnegative edge distances satisfying the triangle inequality, the goal is to find a shortest route that visits all the nodes. The TSP is one of the most fundamental and well-studied problems in combinatorial optimization (with whole books devoted to it \cite{Book1,Book2,Book3}), and there are many variants. For example, in the \emph{tour} version the goal is to find a shortest hamiltonian cycle, while in the \emph{path} version we are additionally given two nodes $s$ and $t$ and the goal is to find a shortest hamiltonian path from $s$ to $t$. Along a different dimension, in the \emph{symmetric} version we assume that for every pair of nodes $u$ and $v$ the distance from $u$ to $v$ equals the distance from $v$ to $u$, while in the \emph{asymmetric} version we make no such assumption.

The approximability of the TSP is not well-understood, and there are large gaps in the approximation ratios between the best $\NP$-hardness results and the best approximation algorithms. For all four of the above variants, the best $\NP$-hardness results are due to \cite{PV}: For every constant $\epsilon>0$, it is $\NP$-hard to approximate the symmetric tour and path versions within $220/219-\epsilon$, and it is $\NP$-hard to approximate the asymmetric tour and path versions within $117/116-\epsilon$. The state-of-the-art polynomial-time approximation algorithms achieve approximation ratios $3/2$ for the symmetric tour version \cite{Chr}, $5/3$ for the symmetric path version \cite{Hoo}, and $O(\log n/\log\log n)$ for the asymmetric tour and path versions \cite{AGMOGS,FS} where $n$ is the number of nodes. There are better algorithms for several important restricted classes of distance functions. For example, there have been breakthroughs for shortest path metrics on unweighted undirected graphs: For the symmetric tour version, a $(3/2-\epsilon_0)$-approximation (for some constant $\epsilon_0>0$) is obtained in \cite{OGSS} and a $1.461$-approximation is obtained in \cite{MS}, and for the symmetric path version, a $(5/3-\epsilon_0)$-approximation (for some constant $\epsilon_0>0$) is obtained in \cite{OGSS,AS} and a $1.586$-approximation is obtained in \cite{MS}. The symmetric tour version can be approximated within $1+\epsilon$ (for every constant $\epsilon>0$) for Euclidean distances in constant-dimensional real space \cite{Aro} and for shortest path metrics on weighted undirected planar graphs \cite{AGKKW,Kle}. The asymmetric tour version has a $O(1)$-approximation algorithm for shortest path metrics on weighted directed planar graphs \cite{OGS}.

For problems such as TSP where the known $\NP$-hardness lower bounds on approximation ratios are quite weak, a natural goal is to get stronger lower bounds for restricted classes of algorithms. Such lower bounds can sometimes be construed as evidence that the problem is indeed hard to approximate within the stronger bound, and they may have the additional advantage of being unconditional ($\NP$-hardness lower bounds are conditional on $\p\ne\NP$). For many combinatorial optimization problems, one very general class of algorithms is those that solve some linear programming relaxation of the problem, ``round'' the solution to an integral solution in some way, and derive their approximation guarantee by comparing the value of the rounded solution to the optimal value of the relaxation. The integrality gap of an instance is the ratio between the optimal values of the original problem and the relaxation. The existence of a family of instances with integrality gap at least $\alpha>1$ proves that no such rounding algorithm can achieve approximation ratio better than $\alpha$. Thus exhibiting instances with large integrality gaps constitutes an unconditional lower bound against this class of algorithms. The integrality gap of a relaxation for a problem is defined to be the maximum integrality gap over all instances (expressed as a function of instance size).

For the above variants of the traveling salesperson problem, there is a classic and well-studied linear programming relaxation which has been called the Dantzig-Fulkerson-Johnson relaxation (after \cite{DFJ}), the Held-Karp relaxation (after \cite{HK1,HK2}), and the subtour elimination relaxation. Since there is no consensus on what to call it, we simply refer to it as the \emph{standard relaxation}. Many of the known approximation algorithms for TSP work by (implicitly or explicitly) rounding solutions to the standard relaxation. For the symmetric tour version, the integrality gap is known to be at least $4/3$. It is conjectured to be exactly $4/3$, and proving this is a notorious open problem in combinatorial optimization. The best upper bound on the integrality gap is $3/2$ \cite{Wol,SW}, though in the case of shortest path metrics on unweighted undirected graphs the upper bound has been improved to $3/2-\epsilon_0$ (for some constant $\epsilon_0>0$) \cite{OGSS} and further to $1.461$ \cite{MS}. For the symmetric path version, the integrality gap is known to be at least $3/2$ and at most $5/3$ \cite{AS}, though in the case of shortest path metrics on unweighted undirected graphs the upper bound has been improved to $5/3-\epsilon_0$ (for some constant $\epsilon_0>0$) \cite{OGSS,AS} and further to essentially $1.586$ \cite{MS}. For the asymmetric tour version, the integrality gap is known to be at least $2$ \cite{CGK} and at most $O(\log n/\log\log n)$ \cite{AGMOGS}. For the asymmetric path version, the integrality gap is known to be at least $2$ \cite{CGK} and at most $O(\log n)$ \cite{FSS}.

Although integrality gap lower bounds rule out algorithms that derive their approximation guarantees by comparing the output to the optimum of a particular relaxation of the problem, one can always construct tighter and tighter relaxations in hope of reducing the integrality gap and getting improved approximation guarantees. Hence it is desirable to prove lower bounds for large \emph{classes} of relaxations, since this constitutes stronger evidence for the hardness of the problem. Lov{\'a}sz and Schrijver \cite{LS} and Sherali and Adams \cite{SA} introduced powerful and general techniques for constructing hierarchies of tighter and tighter relaxations for any problem (see \cite{Tul} for a survey). These techniques are called ``lift-and-project'' procedures, and there has been substantial interest in proving integrality gap lower bounds for various relaxations produced by these procedures (a small sample of results includes \cite{AAT,STT,GMPT,CMM,GMT}). A nice feature of such results is that lower bounds for sufficiently many rounds of the procedure allow us to rule out slightly-subexponential time rounding algorithms, something that PCP-based results often fail to do.

We consider the task of proving integrality gap lower bounds for lift-and-project procedures applied to the standard relaxation of the TSP. There is only one previous paper on this topic: Cheung \cite{Che} proved an essentially tight result for the Lov{\'a}sz-Schrijver procedure applied to the symmetric TSP tour problem --- he showed that the integrality gap remains at least $4/3$ after $o(n)$ rounds of the procedure for $n$-node instances. A natural question is whether the integrality gap lower bound of $2$ for the asymmetric TSP tour problem \cite{CGK} survives these lift-and-project procedures. Our main theorem is a positive result in this direction: We prove that the integrality gap is at least $3/2$ after one round of the Lov{\'a}sz-Schrijver procedure (which is equivalent to one round of the Sherali-Adams procedure). There is a small caveat though: There are two versions of the standard relaxation which are both widely used in the literature and which are equivalent in a certain sense, but one seems to become somewhat weaker than the other after lift-and-project procedures are applied. We can only prove our lower bound for the weaker version (see Section~\ref{tsp:sec:intro:results} for details).

One key challenge in proving Lov{\'a}sz-Schrijver integrality gap lower bounds is in designing so-called protection matrices. None of the previous techniques for designing protection matrices (for problems such as vertex cover, constraint satisfaction problems, and symmetric TSP) seem to help with asymmetric TSP. We introduce a new technique for our setting, based on finding certain combinatorial objects we call \emph{frames} in directed graphs.

Another natural question is whether the folklore integrality gap lower bound of $3/2$ for the symmetric TSP path problem survives lift-and-project procedures. We show that it does, in fact, survive $o(n)$ rounds of the Lov{\'a}sz-Schrijver procedure, by giving a simple reduction to Cheung's result \cite{Che}. The reduction is not generic, but rather exploits specific properties of Cheung's family of instances. This result can be considered evidence that the symmetric TSP path problem cannot be approximated better than $3/2$ (even by slightly-subexponential time algorithms). The fact that we get a lower bound of $3/2$ in both this result and our result for the asymmetric TSP tour problem is a coincidence; the techniques used are quite different.


\subsection{Definitions}
\label{tsp:sec:intro:definitions}

In Section~\ref{tsp:sec:intro:definitions:gap} we give the definitions of integrality gaps and the Lov{\'a}sz-Schrijver lift-and-project procedure. In Section~\ref{tsp:sec:intro:definitions:tsp} we give the definitions of the relevant variants of the traveling salesperson problem and their standard linear programming relaxations.


\subsubsection{Lift-and-Project Integrality Gaps}
\label{tsp:sec:intro:definitions:gap}

For our purpose a \emph{(combinatorial optimization) problem} is a set of instances, where each instance has a set of points $I\subseteq\{0,1\}^m$ and a nonnegative convex objective function $f$, and the goal is to minimize $f$ over points in $I$. In the notation, we suppress the dependence of $I$, $m$, and $f$ on the instance. The instances of a problem may have some size parameter $n$ (not necessarily equal to $m$). A \emph{relaxation} of a problem associates to each instance a convex set $R\subseteq[0,1]^m$ such that $I\subseteq R$. Given a relaxation of a problem, the \emph{integrality gap} of an instance is defined to be $\frac{\min_{x\in I}f(x)}{\min_{x\in R}f(x)}$, and the integrality gap of the relaxation is defined to be the maximum integrality gap of an instance, expressed as a function of the instance size $n$. If the objective function for a problem is always linear, and the set $R$ for a relaxation is always a polytope, then we call the relaxation a \emph{linear programming relaxation}.

Lov{\'a}sz and Schrijver \cite{LS} introduced a general technique for improving linear programming relaxations. Given a convex set $R\subseteq[0,1]^m$, the Lov{\'a}sz-Schrijver procedure produces a new convex set $N(R)\subseteq[0,1]^m$ which is a polytope if $R$ is. The procedure can be applied iteratively, and we use $N^r(R)$ to denote the convex set after $r$ rounds (for example $N^2(R)=N(N(R))$). The $N$ operator has the following properties: $R\supseteq N(R)\supseteq\text{conv}(R\cap\{0,1\}^m)$ and $N^m(R)=\text{conv}(R\cap\{0,1\}^m)$ where $\text{conv}(\cdot)$ denotes convex hull, and thus \[R~\supseteq~N(R)~\supseteq~N^2(R)~\supseteq~\cdots~\supseteq~N^{m-1}(R)~\supseteq~N^m(R)~=~\text{conv}(R\cap\{0,1\}^m).\] Hence the procedure can be applied to a relaxation of a combinatorial optimization problem to yield new relaxations, where the integrality gap is nonincreasing with the number of rounds and drops to $1$ after at most $m$ rounds if $R\cap\{0,1\}^m=I$. Furthermore, if $R$ has a polynomial-time separation oracle, then linear objectives can be optimized over $N^r(R)$ in time $m^{O(r)}$.

We now define the $N$ operator. Given a convex set $R\subseteq[0,1]^m$ we define \[\text{cone}(R)~=~\big\{(\lambda,\lambda x_1,\ldots,\lambda x_m)~:~\text{$\lambda\ge 0$ and $(x_1,\ldots,x_m)\in R$}\big\}\] and we index the $\lambda$ coordinate by $0$. For an $(m+1)\times(m+1)$ matrix $X$ we index the top row and left column by $0$, and we use $X_i$ to denote the $i$th row ($i\in\{0,1,\ldots,m\}$). Given a convex set $R\subseteq[0,1]^m$, a point $x\in[0,1]^m$ is in $N(R)$ if and only if there exists an $(m+1)\times(m+1)$ matrix $X$, called a \emph{protection matrix}, such that
\begin{itemize}
\item[(i)] $X$ is symmetric,
\item[(ii)] $X_0=\text{diag}(X)=(1~~x)$, and
\item[(iii)] $X_i\in\text{cone}(R)$ and $X_0-X_i\in\text{cone}(R)$ for each $i\in\{1,\ldots,m\}$.
\end{itemize}
The procedure is called a ``lift-and-project'' procedure because it first lifts $R$ to a convex set in $[0,1]^{(m+1)^2}$ (namely the set of matrices satisfying (i), (ii), and (iii) for some $x$) which is a polytope if $R$ is, and then projects on certain coordinates to get $N(R)$. Lov{\'a}sz and Schrijver also introduced a stronger operator $N_+$ which is defined in the same way as $N$ except the protection matrix is also required to be positive semidefinite. For all convex sets $R\subseteq[0,1]^m$ we have $N_+(R)\subseteq N(R)$ and \[R~\supseteq~N_+(R)~\supseteq~\cdots~\supseteq~N_+^m(R)~=~\text{conv}(R\cap\{0,1\}^m).\]

Sherali and Adams \cite{SA} introduced a lift-and-project procedure for improving linear programming relaxations which is more powerful than the Lov{\'a}sz-Schrijver procedure, but also more difficult to prove integrality gap lower bounds for. One round of Sherali-Adams coincides with one round of Lov{\'a}sz-Schrijver. We do not define the procedure for higher rounds here.


\subsubsection{The Traveling Salesperson Problem}
\label{tsp:sec:intro:definitions:tsp}

We now formally define the traveling salesperson problem and its standard linear programming relaxation \cite{DFJ,HK1,HK2}. For completeness, we define all four variants discussed above. Given an undirected graph $G=(V,E)$ and a set $S\subseteq V$ we let $\delta(S)$ denote the set of edges crossing the cut $(S,\overline{S})$. For $v\in V$ we define $\delta(v)$ to be $\delta(\{v\})$. Given a directed graph $G=(V,E)$ and a set $S\subseteq V$ we let $\delta^+(S)$ denote the set of edges leaving $S$ and $\delta^-(S)$ denote the set of edges entering $S$. For $v\in V$ we define $\delta^+(v)$ to be $\delta^+(\{v\})$ and $\delta^-(v)$ to be $\delta^-(\{v\})$. For every graph $G=(V,E)$ (undirected or directed), given a set $F\subseteq E$ and a vector $x=(x_e)_{e\in E}$ we define $x(F)=\sum_{e\in F}x_e$. Also, given two vectors $d=(d_e)_{e\in E}$ and $x=(x_e)_{e\in E}$ we define $d\cdot x=\sum_{e\in E}d_ex_e$. We say a directed graph is complete if every pair of nodes has both edges between them and there are no self loops. We use $K_n$ to denote both the complete undirected graph on $n$ nodes and the complete directed graph on $n$ nodes; it will always be clear from context which is meant.

\bigskip\noindent{\bf Symmetric tour version.}
Given an undirected graph $G=(V,E)$, the \emph{symmetric tour polytope} $ST(G)$ has variables $x=(x_e)_{e\in E}$ and is defined by the following constraints.
\begin{alignat*}{2}
x(\delta(S))~&\ge~2&&\hspace{1cm}\forall~\emptyset\subsetneq S\subsetneq V\\
x(\delta(v))~&=~2&&\hspace{1cm}\forall~v\in V\\
x_e~&\in~[0,1]&&\hspace{1cm}\forall~e\in E
\end{alignat*}
We define $ST^\text{int}(G)$ similarly but require $x_e\in\{0,1\}$ for all $e\in E$; note that $ST^\text{int}(G)$ consists of exactly the hamiltonian cycles in $G$. The \emph{symmetric TSP tour problem} is the following: Given a complete undirected graph $K_n=(V,E)$ with nonnegative edge distances $d=(d_e)_{e\in E}$ satisfying the triangle inequality, minimize $d\cdot x$ subject to $x\in ST^\text{int}(K_n)$. The standard relaxation allows $x\in ST(K_n)$.

\bigskip\noindent{\bf Symmetric path version.}
Given an undirected graph $G=(V,E)$ and distinct nodes $s,t\in V$ (the case $s=t$ is covered by the symmetric tour version), the \emph{symmetric path polytope} $SP(G,s,t)$ has variables $x=(x_e)_{e\in E}$ and is defined by the following constraints.
\begin{alignat*}{2}
x(\delta(S))~&\ge~2&&\hspace{1cm}\forall~\emptyset\subsetneq S\subsetneq V\text{ with }\big|S\cap\{s,t\}\big|\ne 1\\
x(\delta(S))~&\ge~1&&\hspace{1cm}\forall~\emptyset\subsetneq S\subsetneq V\text{ with }\big|S\cap\{s,t\}\big|=1\\
x(\delta(v))~&=~2&&\hspace{1cm}\forall~v\in V\backslash\{s,t\}\\
x(\delta(v))~&=~1&&\hspace{1cm}\forall~v\in \{s,t\}\\
x_e~&\in~[0,1]&&\hspace{1cm}\forall~e\in E
\end{alignat*}
We define $SP^\text{int}(G,s,t)$ similarly but require $x_e\in\{0,1\}$ for all $e\in E$; note that $SP^\text{int}(G,s,t)$ consists of exactly the hamiltonian paths in $G$ with $s$ and $t$ as their endpoints. The \emph{symmetric TSP path problem} is the following: Given a complete undirected graph $K_n=(V,E)$ with nonnegative edge distances $d=(d_e)_{e\in E}$ satisfying the triangle inequality, and given distinct nodes $s,t\in V$, minimize $d\cdot x$ subject to $x\in SP^\text{int}(K_n,s,t)$. The standard relaxation allows $x\in SP(K_n,s,t)$.

\newpage

\bigskip\noindent{\bf Asymmetric tour version.}
Given a directed graph $G=(V,E)$, the \emph{asymmetric tour polytope} $AT(G)$ has variables $x=(x_e)_{e\in E}$ and is defined by the following constraints.
\begin{alignat*}{2}
x(\delta^+(S))~&\ge~1&&\hspace{1cm}\forall~\emptyset\subsetneq S\subsetneq V\\
x(\delta^-(S))~&\ge~1&&\hspace{1cm}\forall~\emptyset\subsetneq S\subsetneq V\\
x(\delta^+(v))~&=~1&&\hspace{1cm}\forall~v\in V\\
x(\delta^-(v))~&=~1&&\hspace{1cm}\forall~v\in V\\
x_e~&\in~[0,1]&&\hspace{1cm}\forall~e\in E
\end{alignat*}
Of course, the second group of constraints is redundant given the first, but we prefer to include it for aesthetic reasons. We define $AT^\text{int}(G)$ similarly but require $x_e\in\{0,1\}$ for all $e\in E$; note that $AT^\text{int}(G)$ consists of exactly the hamiltonian cycles in $G$. The \emph{asymmetric TSP tour problem} is the following: Given a complete directed graph $K_n=(V,E)$ with nonnegative edge distances $d=(d_e)_{e\in E}$ satisfying the triangle inequality, minimize $d\cdot x$ subject to $x\in AT^\text{int}(K_n)$. The standard relaxation allows $x\in AT(K_n)$.

\bigskip\noindent{\bf Asymmetric path version.}
Given a directed graph $G=(V,E)$ and distinct nodes $s,t\in V$ (the case $s=t$ is covered by the asymmetric tour version), the \emph{asymmetric path polytope} $AP(G,s,t)$ has variables $x=(x_e)_{e\in E}$ and is defined by the following constraints.
\begin{alignat*}{2}
x(\delta^+(S))~&\ge~1&&\hspace{1cm}\forall~\emptyset\subsetneq S\subseteq V\backslash\{t\}\\
x(\delta^-(S))~&\ge~1&&\hspace{1cm}\forall~\emptyset\subsetneq S\subseteq V\backslash\{s\}\\
x(\delta^+(v))~&=~1&&\hspace{1cm}\forall~v\in V\backslash\{t\}\\
x(\delta^-(v))~&=~1&&\hspace{1cm}\forall~v\in V\backslash\{s\}\\
x(\delta^+(t))~&=~0&&\\
x(\delta^-(s))~&=~0&&\\
x_e~&\in~[0,1]&&\hspace{1cm}\forall~e\in E
\end{alignat*}
The importance of the constraints $x(\delta^+(t))=x(\delta^-(s))=0$ was clarified in \cite{Nag}. We define $AP^\text{int}(G,s,t)$ similarly but require $x_e\in\{0,1\}$ for all $e\in E$; note that $AP^\text{int}(G,s,t)$ consists of exactly the hamiltonian paths in $G$ from $s$ to $t$. The \emph{asymmetric TSP path problem} is the following: Given a complete directed graph $K_n=(V,E)$ with nonnegative edge distances $d=(d_e)_{e\in E}$ satisfying the triangle inequality, and given distinct nodes $s,t\in V$, minimize $d\cdot x$ subject to $x\in AP^\text{int}(K_n,s,t)$. The standard relaxation allows $x\in AP(K_n,s,t)$.

\bigskip Although the cut constraints in the four polytopes $ST$, $SP$, $AT$, $AP$ look rather different from each other, they are actually very uniform: Every cut should have capacity at least $2$ (symmetric versions) or at least $1$ in both directions (asymmetric versions), unless the cut separates $s$ and $t$ in the path versions, in which case it should have capacity at least $1$ (symmetric version) or at least $1$ from $s$ to $t$ (asymmetric version). Note that although all four polytopes have exponentially many constraints, they each have a polynomial-time separation oracle using min-cut computations. Hence linear objectives can be optimized over these polytopes in polynomial time using the ellipsoid algorithm.


\subsection{Results}
\label{tsp:sec:intro:results}

There is only one previous result on lift-and-project integrality gaps for the traveling salesperson problem: Cheung \cite{Che} proved that the integrality gap of the standard relaxation of the symmetric TSP tour problem remains at least $4/3-o(1)$ even after $o(n)$ rounds of the Lov{\'a}sz-Schrijver $N_+$ procedure. We consider whether a lower bound better than $4/3$ can be obtained for the \emph{asymmetric} TSP tour problem. Although we cannot prove a better lower bound for the standard relaxation, we can prove a new result for the following relaxation, which is equivalent to the standard relaxation in a certain sense.

Given a directed graph $G=(V,E)$, the \emph{balanced asymmetric tour polytope} $AT_\text{bal}(G)$ has variables $x=(x_e)_{e\in E}$ and is defined by the following constraints.
\begin{alignat*}{2}
x(\delta^+(S))~&\ge~1&&\hspace{1cm}\forall~\emptyset\subsetneq S\subsetneq V\\
x(\delta^-(S))~&\ge~1&&\hspace{1cm}\forall~\emptyset\subsetneq S\subsetneq V\\
x(\delta^+(v))~&=~x(\delta^-(v))&&\hspace{1cm}\forall~v\in V\\
x_e~&\in~[0,1]&&\hspace{1cm}\forall~e\in E
\end{alignat*}
As before, the second group of constraints is redundant given the first, but we prefer to include it for aesthetic reasons. We define $AT_\text{bal}^\text{int}(G)$ similarly but require $x_e\in\{0,1\}$ for all $e\in E$.

We have $AT(G)\subseteq AT_\text{bal}(G)$ and thus the balanced asymmetric tour polytope yields a relaxation for the asymmetric TSP tour problem, which we call the \emph{balanced standard relaxation}. It is well-known that $AT_\text{bal}(K_n)$ is equivalent to $AT(K_n)$, and $AT_\text{bal}^\text{int}(K_n)$ is equivalent to $AT^\text{int}(K_n)$, in the sense that they have the same minimum values under any objective $d\cdot x$ where $d$ satisfies the triangle inequality (see for example \cite{CGK,Ngu}). For this reason, both formulations are very commonly used in the literature and are considered interchangeable. However, they are not necessarily interchangeable after lift-and-project procedures are applied; $N(AT(K_n))$ might have a smaller integrality gap than $N(AT_\text{bal}(K_n))$. The integrality gap of $AT_\text{bal}(K_n)$ does drop to $1$ after at most $m=n(n-1)$ rounds of $N$, due to the equivalence of $AT_\text{bal}^\text{int}(K_n)$ and $AT^\text{int}(K_n)$.

\begin{theorem}
\label{tsp:thm:asym}
The integrality gap of one round of the Lov{\'a}sz-Schrijver / Sherali-Adams procedure applied to the balanced standard relaxation of the asymmetric TSP tour problem is at least $3/2-o(1)$.
\end{theorem}

We prove Theorem~\ref{tsp:thm:asym} in Section~\ref{tsp:sec:asym}. Note that the lower bound of $3/2-o(1)$ beats the lower bound of $4/3-o(1)$ that follows trivially from Cheung's result for the symmetric TSP tour problem. There are four deficiencies in Theorem~\ref{tsp:thm:asym} that it would be nice to overcome.
\begin{itemize}
\item[(1)] We only handle a single round of lift-and-project; ideally we would like to handle at least a superconstant number of rounds.
\item[(2)] We only handle the balanced standard relaxation; ideally we would like to handle the standard relaxation as defined in Section~\ref{tsp:sec:intro:definitions:tsp}.
\item[(3)] The integrality gap lower bound is only $3/2-o(1)$; ideally we would like to match the lower bound of $2-o(1)$ due to \cite{CGK} which holds for zero rounds of lift-and-project.
\item[(4)] The lower bound only holds for the Lov{\'a}sz-Schrijver $N$ procedure; ideally we would like to handle the stronger $N_+$ procedure.
\end{itemize}

The instances we use to witness Theorem~\ref{tsp:thm:asym} are the same as the ones constructed in \cite{CGK}, but we need a new analysis showing that a certain fractional solution survives one round of lift-and-project. The heart of the analysis involves finding certain sets of edges, which we call \emph{frames}, in the graphs constructed in \cite{CGK}. We need to find a frame associated with each edge, with the property that for every pair of edges $e_1$ and $e_2$, $e_2$ is in $e_1$'s frame if and only if $e_1$ is in $e_2$'s frame (this property is related to the symmetry requirement for protection matrices).

We now turn to the symmetric TSP path problem. By giving a simple reduction to Cheung's result \cite{Che}, we show that the folklore $3/2-o(1)$ lower bound on the integrality gap of the standard relaxation survives $o(n)$ rounds of the Lov{\'a}sz-Schrijver $N_+$ procedure. The reduction is not a generic reduction, but rather exploits special properties of Cheung's family of instances.

\begin{theorem}
\label{tsp:thm:sym}
The integrality gap of $o(n)$ rounds of the Lov{\'a}sz-Schrijver $N_+$ procedure applied to the standard relaxation of the symmetric TSP path problem is at least $3/2-o(1)$. Moreover, this lower bound holds even for instances that are shortest path metrics on unweighted undirected graphs.
\end{theorem}

We prove Theorem~\ref{tsp:thm:sym} in Section~\ref{tsp:sec:sym}.


\section{Asymmetric TSP Tour Problem}
\label{tsp:sec:asym}

In this section we prove Theorem~\ref{tsp:thm:asym}. For every $\epsilon>0$ we need to construct an instance $d=(d_e)_{e\in E}$ on $K_n=(V,E)$ such that
\begin{equation}
\label{tsp:eq:asymgoal}
\frac{\min_{x\in AT^\text{int}(K_n)}d\cdot x}{\min_{x\in N(AT_\text{bal}(K_n))}d\cdot x}~\ge~3/2-\epsilon.
\end{equation}
The instances we use are the same ones constructed by Charikar, Goemans, and Karloff \cite{CGK}. For integers $k\ge 1$ and $r\ge 2$ they construct a directed graph $G_{k,r}$ with nonnegative edge weights as follows (see the illustrations in Figures~\ref{tsp:fig:cgk1}, \ref{tsp:fig:cgk2}, and~\ref{tsp:fig:cgk3}). The graph $G_{1,r}$ consists of two paths of $r+1$ edges on the same $r+2$ nodes, going in opposite directions, and all edges have weight $1$. One of the two endpoints is designated as the source $s$ and the other as the sink $t$. For $k>1$ the graph $G_{k,r}$ consists of $r$ copies of $G_{k-1,r}$ and an additional source node $s$ and sink node $t$, with a path from $s$ to $t$ of $r+1$ edges visiting the sources of the copies in some order, and another path from $t$ to $s$ of $r+1$ edges visiting the sinks of the copies in the opposite order. All the new edges have weight $r^{k-1}$.

Next, as in \cite{CGK} we define a directed graph $L_{k,r}=(V_{k,r},E_{k,r})$ with nonnegative edge weights as follows. Suppose $(s,v_1),(v_2,t),(t,v_3),(v_4,s)$ are the edges incident to $s$ and $t$ in $G_{k,r}$. Then $L_{k,r}$ is defined by taking $G_{k,r}$, removing $s$ and $t$, and including edges $(v_2,v_1)$ and $(v_4,v_3)$, both of weight $r^{k-1}$.

Now fix some $k\ge 1$ and $r\ge 2$ and let $n=|V_{k,r}|=\Theta(r^k)$. We define the edge distances $d=(d_e)_{e\in E}$ on the complete directed graph $K_n=(V_{k,r},E)$ as the shortest path distances in the weighted graph $L_{k,r}$. Note that $d$ is nonnegative and satisfies the triangle inequality, so this is a valid instance.

\begin{lemma}[\cite{CGK}]
\label{tsp:lem:asymint}
For $k\ge 2$ and $r\ge 3$ we have $\min_{x\in AT^\text{int}(K_n)}d\cdot x\ge(2k-1)(r-1)r^{k-1}$.
\end{lemma}

\begin{lemma}
\label{tsp:lem:asymfrac}
For $k\ge 1$ and $r\ge 2$ we have $\min_{x\in N(AT_\text{bal}(K_n))}d\cdot x\le\frac{4}{3}k(r+1)r^{k-1}$.
\end{lemma}

\begin{figure}[t]
\centering
\includegraphics[scale=0.9]{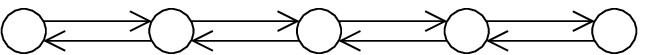}
\caption{The graph $G_{1,3}$.}
\label{tsp:fig:cgk1}
\end{figure}

\begin{figure}[t]
\centering
\includegraphics[scale=0.9]{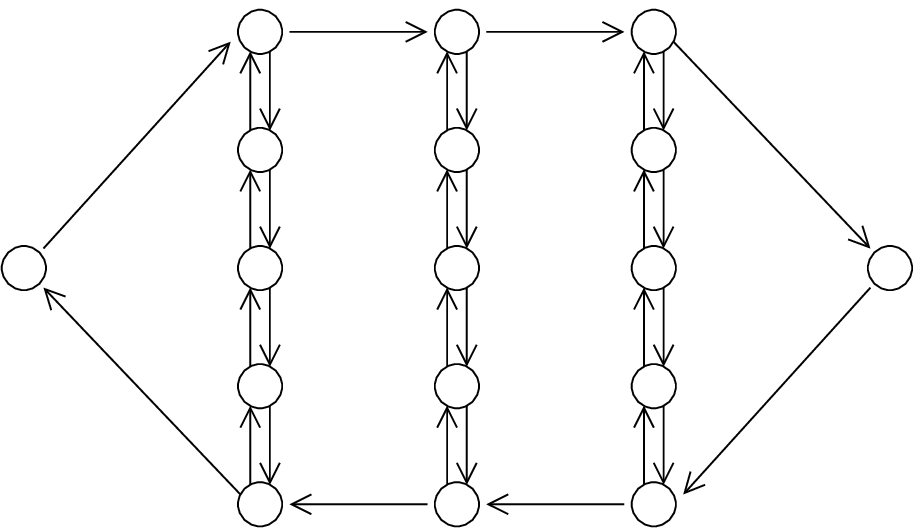}
\caption{The graph $G_{2,3}$.}
\label{tsp:fig:cgk2}
\end{figure}

\begin{figure}[t]
\centering
\includegraphics[scale=0.67]{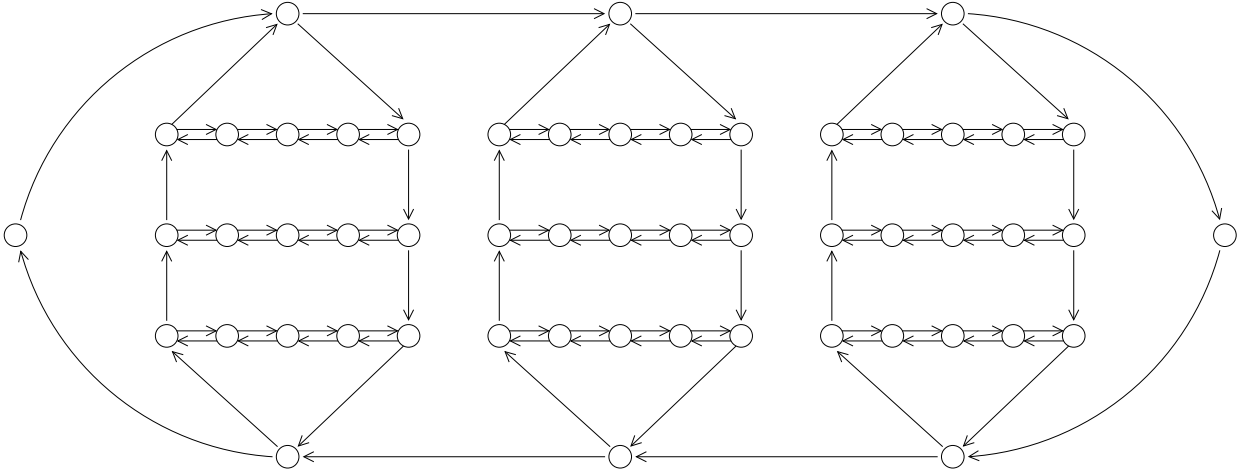}
\caption{The graph $G_{3,3}$.}
\label{tsp:fig:cgk3}
\end{figure}

Letting $k\ge 2$ and $r\ge 3$ and combining Lemma~\ref{tsp:lem:asymint} and Lemma~\ref{tsp:lem:asymfrac}, we find that the ratio on the left side of (\ref{tsp:eq:asymgoal}) is at least $\frac{3}{2}\cdot\frac{(k-1/2)(r-1)}{k(r+1)}$, which is at least $3/2-\epsilon$ provided $k$ and $r$ are large enough. It remains to prove Lemma~\ref{tsp:lem:asymfrac}.

\begin{proof}[Proof of Lemma~\ref{tsp:lem:asymfrac}]
Consider the vector $x=(x_e)_{e\in E_{k,r}}$ with $x_e=2/3$ for all $e\in E_{k,r}$. We claim that $x\in N(AT_\text{bal}(L_{k,r}))$. Then a simple argument shows that extending $x$ with $0$ values for all the edges in $E\backslash E_{k,r}$ yields a point $\hat{x}\in N(AT_\text{bal}(K_n))$ such that $d\cdot\hat{x}$ equals $2/3$ times the sum of all the edge weights in $L_{k,r}$ (using the fact that every edge in the weighted graph $L_{k,r}$ is a shortest path from its tail to its head). An inductive argument shows that the sum of all the edge weights in $L_{k,r}$ is at most $2k(r+1)r^{k-1}$ (see \cite{CGK}), so $\hat{x}$ witnesses Lemma~\ref{tsp:lem:asymfrac}.

All that remains is to prove that $x\in N(AT_\text{bal}(L_{k,r}))$, and this is the meat of the argument. In order to exhibit a protection matrix for $x$, we first need the following definitions. We call a path or cycle in a directed graph \emph{edge-simple} if it does not repeat any edges, but may repeat nodes.

\begin{definition}
\label{tsp:def:frame}
Given a directed graph $G'=(V',E')$ and an edge $(u,v)\in E'$, a \emph{frame} for $(u,v)$ is a set of edges $F\subseteq E'\backslash\{(u,v)\}$ that consists of an edge-simple path from $u$ to $v$ together with zero or more edge-simple cycles. The path and cycles are required to be edge-disjoint from each other.
\end{definition}

\begin{claim}
\label{tsp:clm:frame}
In the graph $L_{k,r}$ there exists a frame $F_e$ for each $e\in E_{k,r}$, with the property that for all $e_1,e_2\in E_{k,r}$ we have $e_2\in F_{e_1}$ if and only if $e_1\in F_{e_2}$.
\end{claim}

We prove Claim~\ref{tsp:clm:frame} shortly, but let us now see how to use it to construct a protection matrix $X$ for $x$. The rows and columns of $X$ are indexed by $E_{k,r}$, except there is an additional $0$th row and $0$th column. We must have $X_{0,0}=1$ and $X_{0,e}=X_{e,0}=X_{e,e}=2/3$ for all $e\in E_{k,r}$. For $e_1\ne e_2$ we set \[X_{e_1,e_2}~=~\begin{cases}1/3&\text{if $e_2\in F_{e_1}$}\\ 1/2&\text{otherwise}\end{cases}.\] Note that $X$ is symmetric by the property in Claim~\ref{tsp:clm:frame}. Thus we just need to show that $X_e\in\text{cone}(AT_\text{bal}(L_{k,r}))$ and $X_0-X_e\in\text{cone}(AT_\text{bal}(L_{k,r}))$ for all $e\in E_{k,r}$.

We first consider $X_e$. Since $(X_e)_0=2/3$, $X_e\in\text{cone}(AT_\text{bal}(L_{k,r}))$ is equivalent to $x^{(e)}\in AT_\text{bal}(L_{k,r})$ where $x^{(e)}$ is the vector on $E_{k,r}$ that gives value $1$ to $e$, value $1/2$ to the edges in $F_e$, and value $3/4$ to the remaining edges. To verify the balance condition $x^{(e)}(\delta^+(v))=x^{(e)}(\delta^-(v))$ for all $v\in V_{k,r}$, consider starting with the vector that assigns $3/4$ to all edges. Since every node in $L_{k,r}$ has two incoming edges and two outgoing edges, this vector would satisfy the balance condition. Reducing the values on the path and cycles to $1/2$ maintains the balance condition except at the endpoints of $e$, which are remedied by raising the value of $e$ to $1$. To verify the cut constraints, suppose we interpret $x^{(e)}$ as capacities in a flow network. Then we need to show that every cut has capacity at least $1$ in both directions, which is equivalent to showing that a unit of flow can be sent from any node to any other node. It is shown in \cite{CGK} that for every $u,v\in V_{k,r}$ there exist two edge-disjoint paths from $u$ to $v$ in $L_{k,r}$. Since every edge has capacity at least $1/2$ in $x^{(e)}$, we can send a half-unit of flow along each of these two paths. Thus we have shown that $X_e\in\text{cone}(AT_\text{bal}(L_{k,r}))$.

Now we consider $X_0-X_e$. Since $(X_0-X_e)_0=1/3$, $X_0-X_e\in\text{cone}(AT_\text{bal}(L_{k,r}))$ is equivalent to $y^{(e)}\in AT_\text{bal}(L_{k,r})$ where $y^{(e)}$ is the vector on $E_{k,r}$ that gives value $0$ to $e$, value $1$ to the edges in $F_e$, and value $1/2$ to the remaining edges. The balance condition is verified similarly as before: The vector that assigns $1/2$ to all edges would satisfy the balance condition, and raising the values on the path and cycles to $1$ maintains the balance condition except at the endpoints of $e$, which are remedied by reducing the value of $e$ to $0$. To verify the cut constraints, we need to show that a unit of flow can be sent from any node to any other node, respecting the capacities $y^{(e)}$. Every edge has capacity at least $1/2$ \emph{except for $e$ itself}, so if we try to send a half-unit of flow along each of the two edge-disjoint paths from $u$ to $v$, we only run into trouble if one of these paths uses $e$. In that case, the half-unit of flow through $e$ can be re-routed via the path in the frame $F_e$. This path may overlap with the two paths from $u$ to $v$, but since each edge in $F_e$ has capacity $1$, it can accommodate two half-units of flow simultaneously. Thus we can always send a unit of flow from $u$ to $v$, and we have shown that $X_0-X_e\in\text{cone}(AT_\text{bal}(L_{k,r}))$.

\begin{figure}[p]
\centering
\includegraphics[scale=0.8]{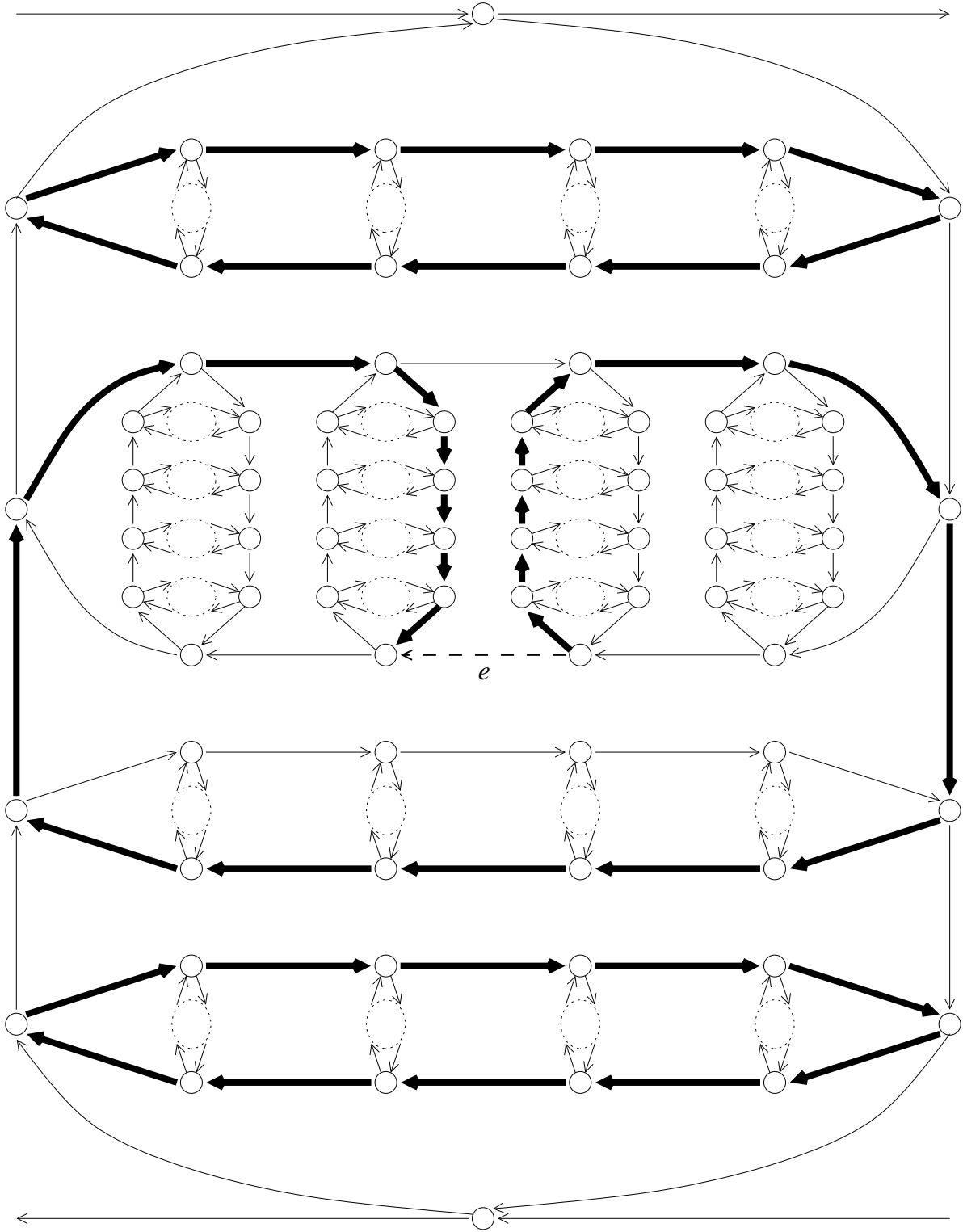}
\caption{Frame for a nonmediating inner edge at level $3\le\ell\le k-2$.}
\label{tsp:fig:frame1}
\end{figure}

\begin{figure}[t]
\centering
\includegraphics[scale=0.67]{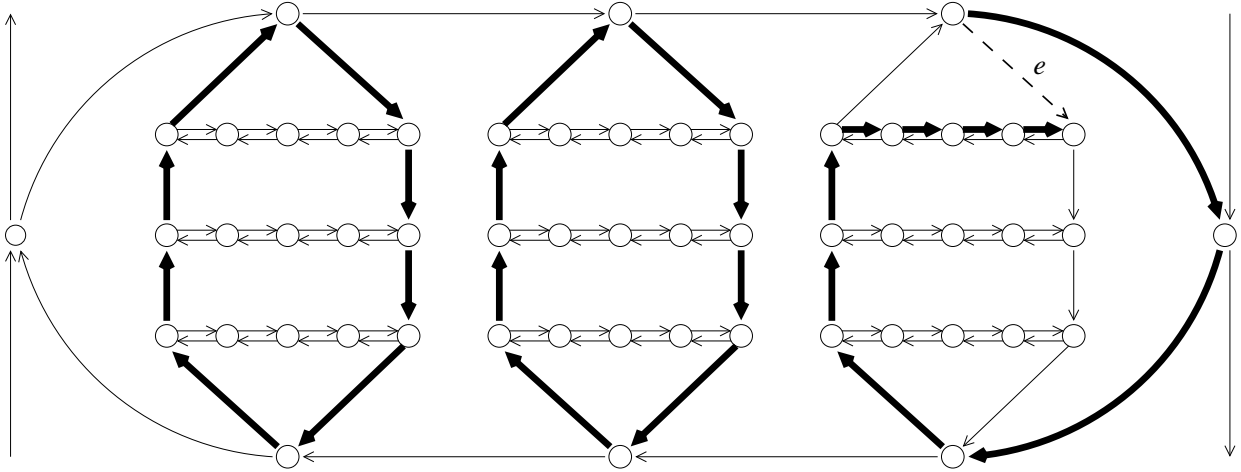}
\caption{Frame for a mediating outer edge at level $\ell=2$.}
\label{tsp:fig:frame2}
\end{figure}

\begin{figure}[t]
\centering
\includegraphics[scale=0.8]{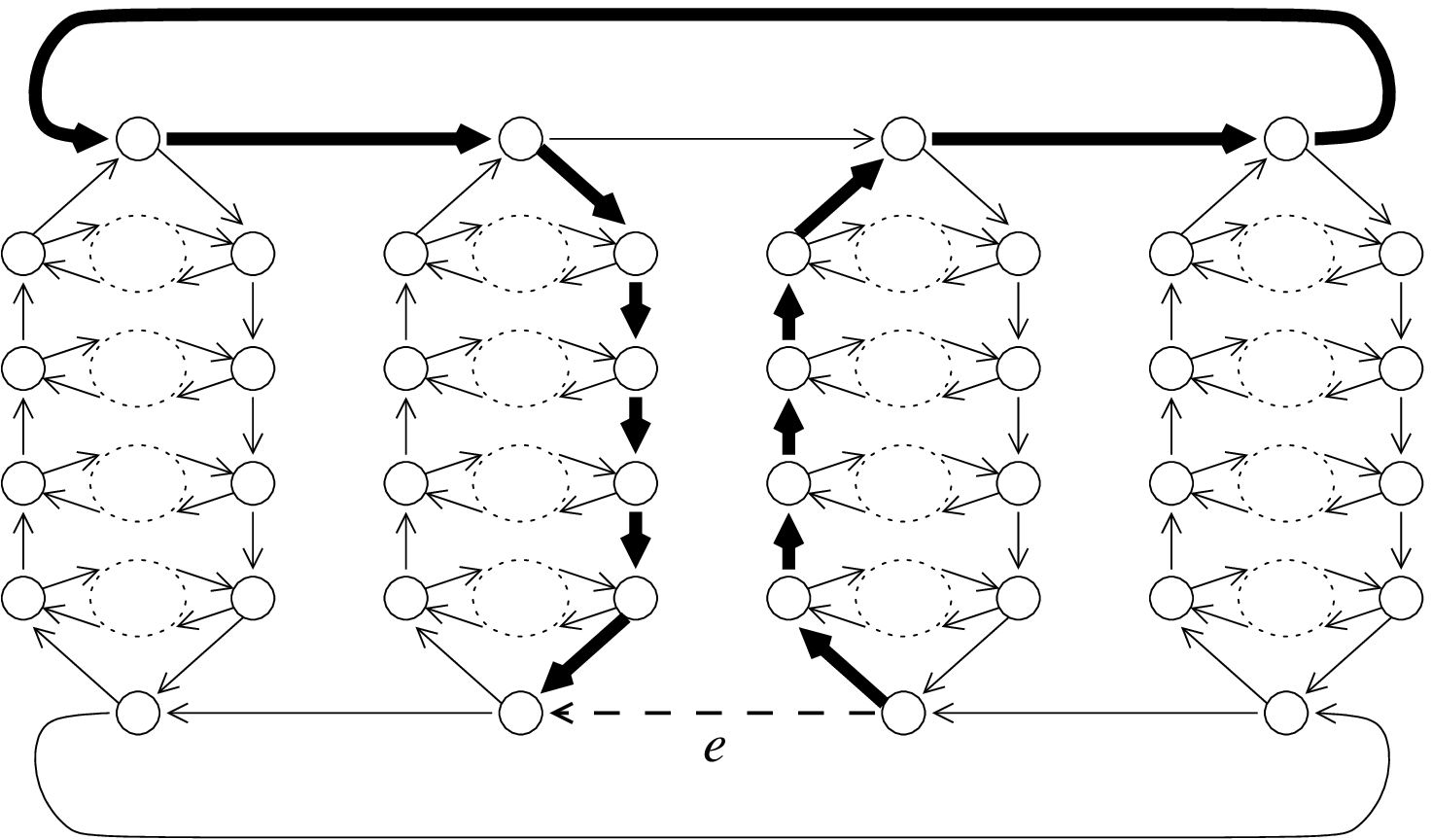}
\caption{Frame for an edge at level $\ell=k$.}
\label{tsp:fig:frame3}
\end{figure}

To finish the proof of Lemma~\ref{tsp:lem:asymfrac}, we just need to prove Claim~\ref{tsp:clm:frame}. For each edge $e\in E_{k,r}$ there is a unique shortest path (in terms of number of edges) from $e$'s tail to $e$'s head that does not use $e$. The natural first attempt at defining the frames is to take $F_e$ to be this path for each $e$. However, this does not work in general; it is crucial that we augment the path with some cycles.

Each edge $e\in E_{k,r}$ can be assigned a level $\ell\in\{1,\ldots,k\}$, meaning that $e$ is introduced as part of a copy of $G_{\ell,r}$ in the construction of $G_{k,r}$ (and the two new edges in $L_{k,r}$ are at level $k$). The most important distinction is whether $\ell<k$ or $\ell=k$. If $\ell<k$ then we include in $F_e$ the shortest path from $e$'s tail to $e$'s head that does not use $e$, and for every copy of $G_{\ell,r}$ that is contained in the ambient copy of $G_{\ell+1,r}$ (or in $L_{k,r}$ if $\ell=k-1$) and that is not touched by this path, we include in $F_e$ a cycle containing all the edges at level $\ell$ in this copy of $G_{\ell,r}$. If $\ell=k$ then we only include in $F_e$ the shortest path from $e$'s tail to $e$'s head that does not use $e$, and we include no cycles. This completes the description of the frames.

Let us be a bit more precise and give some illustrations. From now on, we use ``terminal'' to mean either source or sink, since there is no need to distinguish between the latter. We say an edge at level $\ell<k$ is ``mediating'' if it is incident to a terminal when it is introduced in $G_{\ell,r}$, and ``nonmediating'' otherwise. Also, we say an edge at level $\ell<k-1$ is ``outer'' if when we start at its head and follow the direct route to a terminal of the ambient copy of $G_{\ell,r}$, the other edge entering that terminal is a mediating edge. The edge is an ``inner'' edge otherwise, and all edges at level $k-1$ are considered inner.

Consider an edge $e=(u,v)$ at level $\ell<k$. There are four cases depending on whether $e$ is mediating/nonmediating and outer/inner. The ``typical'' case is when $e$ is nonmediating and inner and at level $3\le\ell\le k-2$. This is illustrated in Figure~\ref{tsp:fig:frame1}, with bold frame edges. Then $u$ is a terminal of a copy of $G_{\ell-1,r}$, and the path starts by taking the direct route to the other terminal of this copy, then the direct route to a terminal of the ambient copy of $G_{\ell,r}$, then one step to the next copy of $G_{\ell,r}$, then the direct route to the other terminal of that copy, then one step back to the copy of $G_{\ell,r}$ containing $e$, then the direct route to the other terminal of the copy of $G_{\ell-1,r}$ of which $v$ is a terminal, and finally the direct route to $v$ through this copy of $G_{\ell-1,r}$. This path involves two copies of $G_{\ell,r}$. For each of the other $r-2$ copies of $G_{\ell,r}$ in the ambient copy of $G_{\ell+1,r}$, we include a cycle containing all the edges at level $\ell$ in this copy of $G_{\ell,r}$. This completes the description of the frame $F_e$. Nothing changes if $\ell=2$ or $\ell=k-1$. When $\ell=1$, there are no copies of ``$G_{\ell-1,r}$'' to go through, so the path simply starts by going directly to a terminal of the copy of $G_{1,r}$ containing $e$, and ends by going from the other terminal directly to $v$.

When $e$ is mediating, all that changes is that one of the two copies of $G_{\ell-1,r}$ disappears; either $u$ or $v$ is already a terminal of the ambient copy of $G_{\ell,r}$. When $e$ is outer, there is no ``next'' copy of $G_{\ell,r}$; the first step along an edge at level $\ell+1$ takes us to a terminal of the ambient copy of $G_{\ell+1,r}$, then the next step takes us back to the copy of $G_{\ell,r}$ containing $e$. Furthermore when $e$ is outer, $r-1$ cycles need to be included in the frame (rather than $r-2$) since the path only touches one copy of $G_{\ell,r}$. Figure~\ref{tsp:fig:frame2} illustrates the frame for a mediating outer edge at level $\ell=2$.

Figure~\ref{tsp:fig:frame3} illustrates the frame for an edge $e=(u,v)$ at level $\ell=k$. The path takes the direct route from $u$ to the other terminal of the copy of $G_{k-1,r}$ containing $u$, then follows the level-$k$ cycle containing this terminal until it gets to the copy of $G_{k-1,r}$ containing $v$, then follows the direct route to $v$ through this copy. This completes the more precise description of all the frames.

It remains to be verified that for all $e_1,e_2\in E_{k,r}$ we have $e_2\in F_{e_1}$ if and only if $e_1\in F_{e_2}$. Of course, we just need to consider an arbitrary $e_1$ and show that for each $e_2\in F_{e_1}$ we have $e_1\in F_{e_2}$. This is a slightly tedious case analysis. Rather than give full details (which would be pedantic) we give a couple illustrative cases. Suppose $e_1$ is the edge $e$ in Figure~\ref{tsp:fig:frame1}. If $e_2$ is the third edge on the path in $F_{e_1}$, then the path in $F_{e_2}$ first goes right to the other terminal in the copy of $G_{\ell-2,r}$, then goes down to $e_1$'s tail, then traverses $e_1$, then goes up through the copy of $G_{\ell-1,r}$ containing $e_1$'s head, and so on. If $e_2$ is the edge that appears right below $e_1$ on the path in $F_{e_1}$, then the path in $F_{e_2}$ only touches the bottom two copies of $G_{\ell,r}$ in the figure, and $e_1$ is contained in a cycle in $F_{e_2}$ (the one for the copy of $G_{\ell,r}$ containing $e_1$). This demonstrates the importance of including cycles in the frames. The other cases can be checked similarly.

This finishes the proof of Claim~\ref{tsp:clm:frame} and the proof of Lemma~\ref{tsp:lem:asymfrac}.
\end{proof}


\section{Symmetric TSP Path Problem}
\label{tsp:sec:sym}

In this section we prove Theorem~\ref{tsp:thm:sym}. For every $\epsilon>0$ we need to construct an instance $d=(d_e)_{e\in E}$ on $K_n=(V,E)$ with $s,t\in V$ such that for some $r=\gamma n$ with $\gamma>0$ depending on $\epsilon$,

\newpage

\begin{equation}
\label{tsp:eq:symgoal}
\frac{\min_{x\in SP^\text{int}(K_n,s,t)}d\cdot x}{\min_{x\in N_+^r(SP(K_n,s,t))}d\cdot x}~\ge~3/2-\epsilon.
\end{equation}
For integers $\ell\ge 1$ and $q\ge 0$ we construct an undirected graph $G_{\ell,q}=(V_{\ell,q},E_{\ell,q})$ as follows (see the illustration in Figure~\ref{tsp:fig:sympath}). The graph $G_{\ell,q}$ has two horizontal node-disjoint paths each with $\ell$ edges. It also has two cliques each with $3q+3$ nodes, a left one and a right one. The cliques are node-disjoint from each other and from the paths. The two left endpoints of the paths each have edges to all the nodes in the left clique, and the two right endpoints of the paths each have edges to all the nodes in the right clique. Finally, there is an additional node $s$ with edges to all the nodes in the left clique, and an additional node $t$ with edges to all the nodes in the right clique.

Also consider the graph $G'_{\ell,q}=(V'_{\ell,q},E'_{\ell,q})$ (illustrated in Figure~\ref{tsp:fig:symtour}) which is the same as $G_{\ell,q}$ except that $s$ and $t$ are connected by a path with $\ell$ new edges (using $\ell-1$ new nodes). We have $V_{\ell,q}\subseteq V'_{\ell,q}$ and $E_{\ell,q}\subseteq E'_{\ell,q}$.

\begin{lemma}
\label{tsp:lem:pathtour}
For all $r\ge 0$, $\ell\ge 1$, and $q\ge 0$ the following holds. For all $x=(x_e)_{e\in E_{\ell,q}}$, if $x'=(x'_e)_{e\in E'_{\ell,q}}$ is the same as $x$ but with all new edges $e$ in $E'_{\ell,q}$ having $x'_e=1$, then $x'\in N_+^r(ST(G'_{\ell,q}))$ implies $x\in N_+^r(SP(G_{\ell,q},s,t))$.
\end{lemma}

\begin{proof}
We prove the lemma by induction on $r$. Suppose $r=0$. Trivially, $x'(\delta_{G'_{\ell,q}}(v))=2$ for all $v\in V_{\ell,q}$ implies $x(\delta_{G_{\ell,q}}(v))=2$ for all $v\in V_{\ell,q}\backslash\{s,t\}$ and $x(\delta_{G_{\ell,q}}(v))=1$ for all $v\in\{s,t\}$. Now consider an arbitrary set $\emptyset\subsetneq S\subsetneq V_{\ell,q}$ with $\big|S\cap\{s,t\}\big|\ne 1$, and assume without loss of generality that $\big|S\cap\{s,t\}\big|=0$. Then $x'(\delta_{G'_{\ell,q}}(S))\ge 2$ implies $x(\delta_{G_{\ell,q}}(S))\ge 2$ since none of the new edges is in $\delta_{G'_{\ell,q}}(S)$. Now consider an arbitrary set $\emptyset\subsetneq S\subsetneq V_{\ell,q}$ with $\big|S\cap\{s,t\}\big|=1$. Then $x'(\delta_{G'_{\ell,q}}(S))\ge 2$ implies $x(\delta_{G_{\ell,q}}(S))\ge 1$ since only one new edge is in $\delta_{G'_{\ell,q}}(S)$. Thus we have shown that $x'\in ST(G'_{\ell,q})$ implies $x\in SP(G_{\ell,q},s,t)$.

Now suppose $r>0$ and the lemma holds for $r-1$. Assume $x'\in N_+(N_+^{r-1}(ST(G'_{\ell,q})))$ and thus it has a protection matrix $X'$ (which has a $0$th row and a $0$th column, and the remaining rows and columns are indexed by $E'_{\ell,q}$). Obtain a matrix $X$ by deleting the rows and columns corresponding to the new edges in $E'_{\ell,q}$. We claim that $X$ is a protection matrix witnessing that $x\in N_+(N_+^{r-1}(SP(G_{\ell,q},s,t)))$. Note that $X$ is symmetric and positive semidefinite since it is a principal submatrix of the symmetric positive semidefinite matrix $X'$. We also have $X_0=\text{diag}(X)=(1~~x)$. We just need to verify that $X_{e^*}\in\text{cone}(N_+^{r-1}(SP(G_{\ell,q},s,t)))$ and $X_0-X_{e^*}\in\text{cone}(N_+^{r-1}(SP(G_{\ell,q},s,t)))$ for each $e^*\in E_{\ell,q}$. We just consider $X_{e^*}$; the case of $X_0-X_{e^*}$ is similar. We know that $X'_{e^*}\in\text{cone}(N_+^{r-1}(ST(G'_{\ell,q})))$. Thus if $X'_{e^*,0}=0$ then $X'_{e^*}$ is all $0$'s, so $X_{e^*}$ is all $0$'s and we have $X_{e^*}\in\text{cone}(N_+^{r-1}(SP(G_{\ell,q},s,t)))$. Otherwise $X'_{e^*,0}>0$ and we have $x'^{(e^*)}\in N_+^{r-1}(ST(G'_{\ell,q}))$ where $x'^{(e^*)}$ is $X'_{e^*}/X'_{e^*,0}$ with the $0$th entry omitted. For each new edge $e\in E'_{\ell,q}$, since $X'_{0,e}=x'_e=1$ we must have $X'_{e^*,e}=X'_{e^*,0}$, by a basic property of protection matrices. Thus $x'^{(e^*)}_e=1$ for each new edge $e$, and the induction hypothesis tells us that $x^{(e^*)}\in N_+^{r-1}(SP(G_{\ell,q},s,t))$ where $x^{(e^*)}$ is $x'^{(e^*)}$ with the entries for the new edges omitted, or in other words, $x^{(e^*)}$ is $X_{e^*}/X_{e^*,0}$ with the $0$th entry omitted. Thus we have $X_{e^*}\in\text{cone}(N_+^{r-1}(SP(G_{\ell,q},s,t)))$ as desired.
\end{proof}

Now fix some integers $r\ge 0$ and $\ell\ge 1$ and let $q=r$ and $n=|V_{\ell,r}|=2(\ell+1)+2(3r+3)+2$. We define the edge distances $d=(d_e)_{e\in E}$ on the complete undirected graph $K_n=(V_{\ell,r},E)$ as the shortest path distances in the unweighted graph $G_{\ell,r}$. Note that $d$ is nonnegative and satisfies the triangle inequality, so this is a valid instance (together with the distinguished nodes $s,t$).

\begin{figure}[t]
\centering
\includegraphics[scale=0.8]{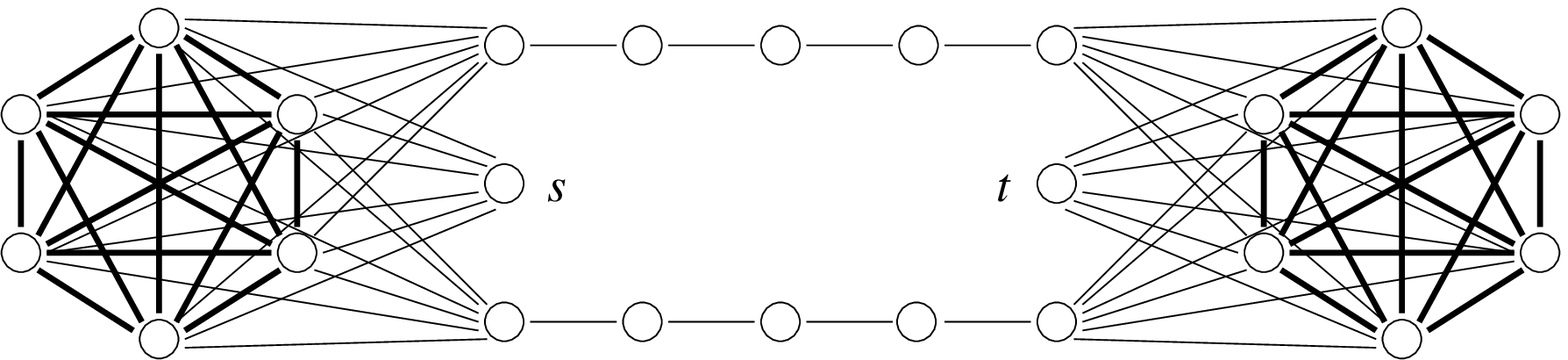}
\caption{The graph $G_{4,1}$.}
\label{tsp:fig:sympath}
\end{figure}

\begin{figure}[t]
\centering
\includegraphics[scale=0.8]{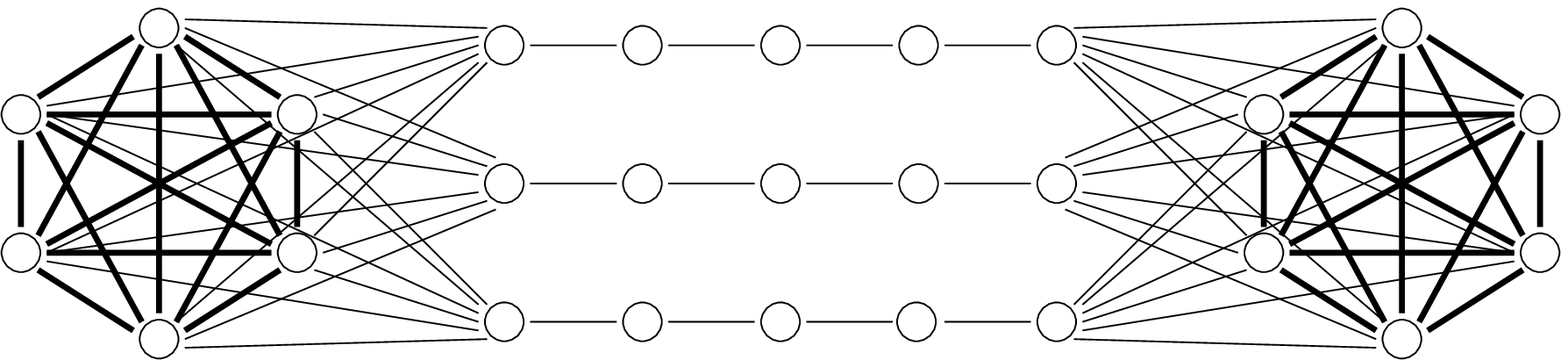}
\caption{The graph $G'_{4,1}$.}
\label{tsp:fig:symtour}
\end{figure}

\begin{lemma}
\label{tsp:lem:symint}
We have $\min_{x\in SP^\text{int}(K_n,s,t)}d\cdot x\ge 3\ell-2$.
\end{lemma}

\begin{lemma}
\label{tsp:lem:symfrac}
We have $\min_{x\in N_+^r(SP(K_n,s,t))}d\cdot x\le 2\ell+6r+9$.
\end{lemma}

Combining Lemma~\ref{tsp:lem:symint} and Lemma~\ref{tsp:lem:symfrac}, we find that the ratio on the left side of (\ref{tsp:eq:symgoal}) is at least $\frac{3\ell-2}{2\ell+6r+9}$, which is at least $3/2-\epsilon$ provided $\ell$ is large enough and $r\le\gamma'\ell$ for some small enough $\gamma'>0$ depending on $\epsilon$. The latter is implied by $r=\gamma n$ for some small enough $\gamma>0$ depending on $\gamma'$. It remains to prove Lemma~\ref{tsp:lem:symint} and Lemma~\ref{tsp:lem:symfrac}.

\begin{proof}[Proof of Lemma~\ref{tsp:lem:symint}]
This is equivalent to lower-bounding the cost of a (non-simple) $s$-$t$ path in $G_{\ell,r}$ that visits each node at least once. Given such a path, first note that at most one of the edges on the top and bottom paths in $G_{\ell,r}$ is not traversed (otherwise some node would not be visited). We claim that for either the top path or the bottom path, all but at most one of the edges are traversed at least twice. Suppose not; then there exist two top edges and two bottom edges, all four of which are traversed at most once. By symmetry, assume the first of these four edges to be traversed is on top; then the $s$-$t$ path cannot return to the left clique without first going to the right clique. Once in the right clique, the path must traverse at least one of the two special bottom edges, but then there is no way to get back to the right clique to end at $t$. This proves the claim. These observations imply that the contribution of the top and bottom paths to the cost of the $s$-$t$ path is at least $3\ell-2$.
\end{proof}

\begin{proof}[Proof of Lemma~\ref{tsp:lem:symfrac}]
Consider the vector $x=(x_e)_{e\in E_{\ell,r}}$ defined as follows. The edges $e$ on the two paths all have $x_e=1$. The edges $e$ within the two cliques all have $x_e=\frac{2-1/(r+1)}{3r+2}$. All remaining edges $e$ have $x_e=1/(3r+3)$ (that is, those edges that go between a clique and an endpoint of one of the paths, or are incident to $s$ or $t$). Let $x'=(x'_e)_{e\in E'_{\ell,r}}$ be the same as $x$ but with all new edges $e$ in $E'_{\ell,r}$ having $x'_e=1$. Cheung \cite{Che} proves that $x'\in N_+^r(ST(G'_{\ell,r}))$; by Lemma~\ref{tsp:lem:pathtour} this implies $x\in N_+^r(SP(G_{\ell,r},s,t))$. Then a simple argument shows that extending $x$ with $0$ values for all the edges in $E\backslash E_{\ell,r}$ yields a point $\hat{x}\in N_+^r(SP(K_n,s,t))$ such that $d\cdot\hat{x}=\sum_{e\in E_{\ell,r}}x_e$ (using the fact that every edge in $G_{\ell,r}$ is the shortest path between its endpoints). We have \[\textstyle\sum_{e\in E_{\ell,r}}x_e~=~2\ell\cdot 1+2\genfrac{(}{)}{0pt}{}{3r+3}{2}\cdot\frac{2-1/(r+1)}{3r+2}+6(3r+3)\cdot\frac{1}{3r+3}~=~2\ell+6r+9\] so $\hat{x}$ witnesses Lemma~\ref{tsp:lem:symfrac}.
\end{proof}


\section{Open Problems}
\label{tsp:sec:open}

There are many open problems on approximability and integrality gaps for variants of the traveling salesperson problem; we give a small sample.

It is open to prove a nontrivial lower bound on the integrality gap of the standard relaxation for the symmetric TSP tour problem after two rounds of the Sherali-Adams procedure. A key technique used in most previous work on Sherali-Adams integrality gap lower bounds is the technique of constructing ``locally consistent distributions'' (introduced in \cite{FVKM,CMM}). However, this technique only seems to apply to relaxations with local constraints, whereas the cut constraints for TSP are very global.

Regarding Theorem~\ref{tsp:thm:asym}, it would be interesting to overcome any of the four deficiencies described in Section~\ref{tsp:sec:intro:results}.

Finally, can the lift-and-project procedures of Lov{\'a}sz-Schrijver or Sherali-Adams be employed to get improved algorithms for any variant of the traveling salesperson problem?


\section*{Acknowledgments}

I thank Luca Trevisan for suggesting to explore lift-and-project procedures applied to the traveling salesperson problem, and I thank Siu On Chan for discussions on lift-and-project procedures.


\bibliographystyle{alpha}
\bibliography{tsp}

\newcommand{\etalchar}[1]{$^{#1}$}
\begin{thebibliography}{FdlVKM07}

\bibitem[AAT05]{AAT}
Michael Alekhnovich, Sanjeev Arora, and Iannis Tourlakis.
\newblock Towards strong nonapproximability results in the
  {L}ov{\'a}sz-{S}chrijver hierarchy.
\newblock In {\em Proceedings of the 37th ACM Symposium on Theory of
  Computing}, pages 294--303, 2005.

\bibitem[ABCC07]{Book2}
David Applegate, Robert Bixby, Vasek Chvatal, and William Cook.
\newblock {\em The Traveling Salesman Problem: A Computational Study}.
\newblock Princeton University Press, 2007.

\bibitem[AGK{\etalchar{+}}98]{AGKKW}
Sanjeev Arora, Michelangelo Grigni, David Karger, Philip Klein, and Andrzej
  Woloszyn.
\newblock A polynomial-time approximation scheme for weighted planar graph
  {TSP}.
\newblock In {\em Proceedings of the 9th ACM-SIAM Symposium on Discrete
  Algorithms}, pages 33--41, 1998.

\bibitem[AGM{\etalchar{+}}10]{AGMOGS}
Arash Asadpour, Michel Goemans, Aleksander Madry, Shayan Oveis~Gharan, and Amin
  Saberi.
\newblock An ${O}(\log n/\log\log n)$-approximation algorithm for the
  asymmetric traveling salesman problem.
\newblock In {\em Proceedings of the 21st ACM-SIAM Symposium on Discrete
  Algorithms}, pages 379--389, 2010.

\bibitem[Aro98]{Aro}
Sanjeev Arora.
\newblock Polynomial time approximation schemes for {E}uclidean traveling
  salesman and other geometric problems.
\newblock {\em Journal of the ACM}, 45(5):753--782, 1998.

\bibitem[AS11]{AS}
Hyung-Chan An and David Shmoys.
\newblock {LP}-based approximation algorithms for traveling salesman path
  problems.
\newblock {\em CoRR}, abs/1105.2391, 2011.

\bibitem[CGK06]{CGK}
Moses Charikar, Michel Goemans, and Howard Karloff.
\newblock On the integrality ratio for the asymmetric traveling salesman
  problem.
\newblock {\em Mathematics of Operations Research}, 31(2):245--252, 2006.

\bibitem[Che05]{Che}
Kevin Cheung.
\newblock On {L}ov{\'a}sz-{S}chrijver lift-and-project procedures on the
  {D}antzig-{F}ulkerson-{J}ohnson relaxation of the {TSP}.
\newblock {\em SIAM Journal on Optimization}, 16(2):380--399, 2005.

\bibitem[Chr76]{Chr}
Nicos Christofides.
\newblock Worst-case analysis of a new heuristic for the travelling salesman
  problem.
\newblock Technical Report 388, Graduate School of Industrial Administration,
  Carnegie Mellon University, 1976.

\bibitem[CMM09]{CMM}
Moses Charikar, Konstantin Makarychev, and Yury Makarychev.
\newblock Integrality gaps for {S}herali-{A}dams relaxations.
\newblock In {\em Proceedings of the 41st ACM Symposium on Theory of
  Computing}, pages 283--292, 2009.

\bibitem[DFJ54]{DFJ}
George Dantzig, Delbert Fulkerson, and Selmer Johnson.
\newblock Solution of a large-scale traveling-salesman problem.
\newblock {\em Journal of the Operations Research Society of America},
  2:393--410, 1954.

\bibitem[FdlVKM07]{FVKM}
Wenceslas Fernandez de~la Vega and Claire Kenyon-Mathieu.
\newblock Linear programming relaxations of maxcut.
\newblock In {\em Proceedings of the 18th ACM-SIAM Symposium on Discrete
  Algorithms}, pages 53--61, 2007.

\bibitem[FS07]{FS}
Uriel Feige and Mohit Singh.
\newblock Improved approximation ratios for traveling salesperson tours and
  paths in directed graphs.
\newblock In {\em Proceedings of the 10th International Workshop on
  Approximation Algorithms for Combinatorial Optimization Problems}, pages
  104--118, 2007.

\bibitem[FSS10]{FSS}
Zachary Friggstad, Mohammad Salavatipour, and Zoya Svitkina.
\newblock Asymmetric traveling salesman path and directed latency problems.
\newblock In {\em Proceedings of the 21st ACM-SIAM Symposium on Discrete
  Algorithms}, pages 419--428, 2010.

\bibitem[GMPT10]{GMPT}
Konstantinos Georgiou, Avner Magen, Toniann Pitassi, and Iannis Tourlakis.
\newblock Integrality gaps of $2-o(1)$ for vertex cover {SDP}s in the
  {L}ov{\'a}sz-{S}chrijver hierarchy.
\newblock {\em SIAM Journal on Computing}, 39(8):3553--3570, 2010.

\bibitem[GMT09]{GMT}
Konstantinos Georgiou, Avner Magen, and Madhur Tulsiani.
\newblock Optimal {S}herali-{A}dams gaps from pairwise independence.
\newblock In {\em Proceedings of the 12th International Workshop on
  Approximation Algorithms for Combinatorial Optimization Problems}, pages
  125--139, 2009.

\bibitem[GP07]{Book3}
Gregory Gutin and Abraham Punnen, editors.
\newblock {\em The Traveling Salesman Problem and Its Variations}.
\newblock Springer, 2007.

\bibitem[HK70]{HK1}
Michael Held and Richard Karp.
\newblock The traveling-salesman problem and minimum spanning trees.
\newblock {\em Operations Research}, 18:1138--1162, 1970.

\bibitem[HK71]{HK2}
Michael Held and Richard Karp.
\newblock The traveling-salesman problem and minimum spanning trees: Part {II}.
\newblock {\em Mathematical Programming}, 1:6--25, 1971.

\bibitem[Hoo91]{Hoo}
Han Hoogeveen.
\newblock Analysis of {C}hristofides’ heuristic: Some paths are more
  difficult than cycles.
\newblock {\em Operations Research Letters}, 10(5):291--295, 1991.

\bibitem[Kle08]{Kle}
Philip Klein.
\newblock A linear-time approximation scheme for {TSP} in undirected planar
  graphs with edge-weights.
\newblock {\em SIAM Journal on Computing}, 37(6):1926--1952, 2008.

\bibitem[LLRKS85]{Book1}
Eugene Lawler, Jan Lenstra, Alexander Rinnooy~Kan, and David Shmoys, editors.
\newblock {\em The Traveling Salesman Problem: A Guided Tour of Combinatorial
  Optimization}.
\newblock Wiley, 1985.

\bibitem[LS91]{LS}
L{\'a}szl{\'o} Lov{\'a}sz and Alexander Schrijver.
\newblock Cones of matrices and set-functions and 0-1 optimization.
\newblock {\em SIAM Journal on Optimization}, 1:166--190, 1991.

\bibitem[MS11]{MS}
Tobias M{\"o}mke and Ola Svensson.
\newblock Approximating graphic {TSP} by matchings.
\newblock {\em CoRR}, abs/1104.3090, 2011.

\bibitem[Nag08]{Nag}
Viswanath Nagarajan.
\newblock On the {LP} relaxation of the asymmetric traveling salesman path
  problem.
\newblock {\em Theory of Computing}, 4(1):191--193, 2008.

\bibitem[Ngu08]{Ngu}
Th{\`a}nh Nguyen.
\newblock A simple {LP} relaxation for the asymmetric traveling salesman
  problem.
\newblock In {\em Proceedings of the 11th International Workshop on
  Approximation Algorithms for Combinatorial Optimization Problems}, pages
  207--218, 2008.

\bibitem[OGS11]{OGS}
Shayan Oveis~Gharan and Amin Saberi.
\newblock The asymmetric traveling salesman problem on graphs with bounded
  genus.
\newblock In {\em Proceedings of the 22nd ACM-SIAM Symposium on Discrete
  Algorithms}, pages 967--975, 2011.

\bibitem[OGSS10]{OGSS}
Shayan Oveis~Gharan, Amin Saberi, and Mohit Singh.
\newblock A randomized rounding approach to the traveling salesman problem.
\newblock Manuscript, December 2010.

\bibitem[PV06]{PV}
Christos Papadimitriou and Santosh Vempala.
\newblock On the approximability of the traveling salesman problem.
\newblock {\em Combinatorica}, 26(1):101--120, 2006.

\bibitem[SA90]{SA}
Hanif Sherali and Warren Adams.
\newblock A hierarchy of relaxations between the continuous and convex hull
  representations for zero-one programming problems.
\newblock {\em SIAM Journal on Discrete Mathematics}, 3(3):411--430, 1990.

\bibitem[STT07]{STT}
Grant Schoenebeck, Luca Trevisan, and Madhur Tulsiani.
\newblock Tight integrality gaps for {L}ov{\'a}sz-{S}chrijver {LP} relaxations
  of vertex cover and max cut.
\newblock In {\em Proceedings of the 39th ACM Symposium on Theory of
  Computing}, pages 302--310, 2007.

\bibitem[SW90]{SW}
David Shmoys and David Williamson.
\newblock Analyzing the {H}eld-{K}arp {TSP} bound: A monotonicity property with
  application.
\newblock {\em Information Processing Letters}, 35(6):281--285, 1990.

\bibitem[Tul11]{Tul}
Madhur Tulsiani.
\newblock {L}ov{\'a}sz-{S}chrijver reformulation.
\newblock In {\em Wiley Encyclopedia of Operations Research and Management
  Science}. Wiley, 2011.

\bibitem[Wol80]{Wol}
Laurence Wolsey.
\newblock Heuristic analysis, linear programming and branch and bound.
\newblock {\em Mathematical Programming Study}, 13:121--134, 1980.

\end{thebibliography}

\end{document}